\documentclass[a4paper,12pt]{article}
\usepackage{amsmath,amsthm,amsfonts,amssymb,bbm}
\usepackage{graphicx,psfrag,subfigure,color,cite}

\numberwithin{equation}{section}

\newcommand{\e}{\varepsilon}
\newcommand{\Pb}{\mathbb{P}}

\newcommand{\R}{\mathbb{R}}

\newcommand{\Z}{\mathbb{Z}}

\newcommand{\Or}{{\cal O}}

\newcommand{\GOE}{\mathrm{GOE}}

\newtheorem{prop}{Proposition}[section]
\newtheorem{thm}[prop]{Theorem}
\newtheorem{lem}[prop]{Lemma}

\newtheorem{cor}[prop]{Corollary}

\newtheorem{cla}[prop]{Claim}

\newtheorem{rem}[prop]{Remark}
\newenvironment{remark}{\begin{rem}\normalfont}{\end{rem}}

\title{Statistics of TASEP \\with three merging characteristics}
\author{Patrik L.\ Ferrari\thanks{Institute for Applied Mathematics, Bonn University, Endenicher Allee 60, 53115 Bonn, Germany. E-mail: {\tt ferrari@uni-bonn.de}} \and
Peter Nejjar\thanks{Institute for Applied Mathematics, Bonn University, Endenicher Allee 60, 53115 Bonn, Germany. E-mail: {\tt nejjar@uni-bonn.de}}
}

\date{\emph{We dedicate this paper to Joel Lebowitz\\ on the occasion of his 90th birthday.}}

\begin{document}
\maketitle

\sloppy
\begin{abstract}
In this paper we consider the totally asymmetric simple exclusion process, with non-random initial condition having three regions of constant densities of particles. From left to right, the densities of the three regions are increasing. Consequently, there are  three characteristics which meet, i.e.\ two shocks  merge. We study the particle fluctuations at this merging point and show that they are given by a product of three (properly scaled) GOE Tracy-Widom distribution functions. We work directly in TASEP  without relying on the connection to last passage percolation.
\end{abstract}

\section{Introduction and main result}
The totally asymmetric simple exclusion process (TASEP) is, arguably, the simplest \emph{non-reversible} interacting stochastic particle system~\cite{Li85b}. It belongs to the Kardar-Parisi-Zhang universality class~\cite{KPZ86}. A TASEP configuration is described by $\eta\in\{0,1\}^\Z$, where $\eta_j$ is the occupation variable at position $j$,  $\eta_j=1$ meaning that  there is a particle at site $j$. The stochastic updating rule is simple: independently of each other, particles try jump to the right with rate $1$ and are allowed to do so only if their right neighboring site is empty. We label particles and denote the position of particle $n$ at time $t$ by $x_n(t)$. Our choice is a right-to-left labeling, i.e., $x_{n+1}(0)<x_n(0)$ for all $n$, which is preserved at all times since particles can not overtake each other.

Under hydrodynamic scaling, the particle density solves the deterministic Burgers equation (see e.g.~\cite{Lig76,AV87}). It is well-known that depending on the initial conditions, discontinuities of the particle density can arise, or persist if already present initially. These discontinuities are called shocks and many properties of them are known. For random (Bernoulli) initial conditions one sees Gaussian fluctuations of the position of the shock~\cite{Fer90,FF94b,PG90} in the $t^{1/2}$ scale. Further microscopic information on the shock are available too~\cite{DJLS93,FKS91,Fer86,BS13}. The Gaussian fluctuations actually comes from the random initial conditions and not from the TASEP dynamics as the dynamical fluctuations generated by models in the KPZ class grow like $t^{1/3}$.

Fluctuations of a single shock for non-random initial conditions have been analyzed the first time in~\cite{FN13} and further investigated in~\cite{FN16,N17}, including the fluctuations of a second class particle~\cite{FGN17}. For large time $t$, the limiting distribution of particle positions around the shock is given by a product of two distribution functions. The reason is that a shock is located in a position where two characteristic lines of the Burgers equation intersect and space-time correlations are non-trivial for points in a $t^{2/3}$-neighborhood of a characteristic line. The $t^{2/3}$ is the typical size of correlation in KPZ growth models.

The shock between two particle densities $\lambda<\rho$ is stable in time and moves with average speed $1-\lambda-\rho$. If in the system there are many shocks they will merge. Coalescing of two shocks is typical, while higher order shock collisions are possible but atypical, because in that case at least three shock trajectories have to meet exactly at a given time. For that reason we consider in this paper the merging of two shocks. When two shocks merge, then three characteristic lines come together, see Figure~\ref{FigShockCollisions}. As stated in Theorem~\ref{ThmMain}, our main result is that the fluctuations of particles when two shocks merge is given as product of three instead of two GOE Tracy-Widom distribution functions.

In the previous papers~\cite{FN13,FN16,N17,FGN17} the results are obtained by using the mapping to the last passage percolation (LPP) model and then analyzing an equivalent problem in terms of LPP quantities. Merging of two shocks corresponds to a LPP from a domain given by the union of three line segments with different slopes. In the LPP approach, one would have to show that three maximizers do not intersect asymptotically.  The LPP is not more difficult than the present approach, but it fails to provide a direct understanding of the shock, and provides no ansatz on how to generalize to PASEP. The space-time picture we present for TASEP in the background of our proof should physically be the same also in models like PASEP or the KPZ equation, for which no correspondence to LPP exists.

This motivates the approach of this paper, which has already been proven crucial to obtain results beyond the totally asymmetric case, see~\cite{N19}. Hence in this paper we work directly in the space-time TASEP picture. As ingredients we use the convergence to a GOE Tracy-Widom distribution for constant density initial condition~\cite{FO17}, the slow-decorrelation phenomenon typical for KPZ models~\cite{Fer08,CFP10b} and show that the relevant randomness for the particle fluctuations at time $t$ are localized around space-time regions of width $t^{2/3+\e}$ around the characteristic lines. The proof that of the localization for space-time has some similarities, but also important differences, with respect to the proof of the localization of the maximizer path in LPP. For LPP, this was analyzed in the work on transversal fluctuations in~\cite{Jo00} and refined in~\cite{BSS14}.

 Of the aforementioned ingredients, in particular the convergence to GOE  in the partially asymmetric simple exclusion process is  not known, whereas slow decorrelation also holds for PASEP.  Furthermore, an important  property which is not satisfied by the PASEP is the  decomposition in several step initial data \eqref{eqMin}: in PASEP, the decomposition \eqref{eqMin} only provides an upper bound for the particle position. For a different shock, this issue - and how to overcome it - is described in detail in  Section 1.2 of \cite{N19}.

\subsubsection*{The main result}
Let us consider the following initial condition, given in terms of three fixed particle densities $0<\rho_1<\rho_2<\rho_3<1$:
\begin{equation}\label{eqIC}
x_n(0)=\left\{
\begin{array}{ll}
-\lfloor n/\rho_1\rfloor,& \textrm{for }n\geq 0,\\
-\lfloor n/\rho_2\rfloor,& \textrm{for }-1\leq n\leq -\lfloor T\rho_2\rfloor,\\
T-\lfloor (n+\lfloor T\rho_2\rfloor)/\rho_3\rfloor,& \textrm{for } n<-\lfloor T\rho_2\rfloor.
\end{array}
\right.
\end{equation}
This means that left from the origin we have density $\rho_1$ of particles, between $0$ and $T$ we have density $\rho_2$ and to the right of $T$ we have density $\rho_3$.

The chosen initial condition implies that there is a shock starting from the origin and moving with speed $1-\rho_1-\rho_2$ and a second shock starting from position $T$ having speed $1-\rho_2-\rho_3$. Thus there is a time where the two shock meet. We call the position and time when this happens the \emph{triple point}, since three characteristic lines of the Burgers equation meet. We consider $T\gg 1$. Then, an elementary computation shows that the triple point occurs around time $T/(\rho_3-\rho_1)$ at position $(1-\rho_1-\rho_2)T/(\rho_3-\rho_1)$. Furthermore, the particle at the triple point has label close to $\rho_1\rho_2 T/(\rho_3-\rho_1)$, see Figure~\ref{FigShockCollisions} for an illustration. Our main result is the following.
\begin{thm}\label{ThmMain}
Consider TASEP with initial condition \eqref{eqIC}. Take
\begin{equation}
N=\frac{\rho_1\rho_2 T}{\rho_3-\rho_1}+u T^{1/3},\quad t=\frac{T}{\rho_3-\rho_1}+\tau T^{1/3},\quad X=\frac{(1-\rho_1-\rho_2)T}{\rho_3-\rho_1}.
\end{equation}
Then
\begin{equation}
\lim_{T\to\infty} \Pb(x_N(t)\geq X-sT^{1/3}) = \prod_{k\in\{1,2,3\}}F_{\rm GOE}((s-\mu_k u+\nu_k\tau)/\sigma_k),
\end{equation}
where
\begin{equation}
\sigma_k = \frac{1}{(\rho_3-\rho_1)^{1/3}}\frac{2^{2/3}\rho_k^{1/3}}{(1-\rho_k)^{2/3}},\quad \mu_k = 1/\rho_k,\quad \nu_k=1-\rho_k,
\end{equation}
for $k=1,2,3$. Here $F_{\rm GOE}$ is the GOE Tracy-Widom distribution function~\cite{TW96}.
\end{thm}

\begin{figure}
\begin{center}
\psfrag{x}[lc]{\textrm{Space}}
\psfrag{t}[lc]{\textrm{Time}}
\psfrag{0}[lc]{$0$}
\psfrag{T}[cc]{$T$}
\psfrag{L1}[l]{${\cal L}_1$}
\psfrag{L2}[l]{${\cal L}_2$}
\psfrag{L3}[l]{${\cal L}_3$}
\includegraphics[height=5cm]{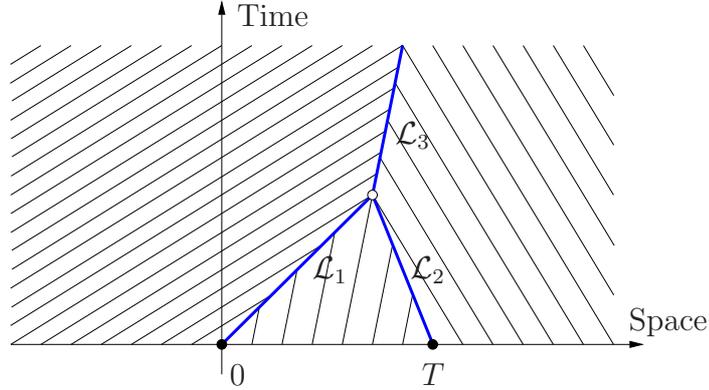}
\caption{Illustration of the space-time characteristics for densities $\rho_1=0.1$, $\rho_2=0.4$ and $\rho_3=0.8$. The black lines are characteristic lines in the system. The blue thick lines  $\mathcal{L}_1,\mathcal{L}_2, \mathcal{L}_3$ are the three standard shocks, while the white dot is the triple point. Along any of the black characteristics, the fluctuations are given by a single $\GOE$ distribution. On each of the lines $\mathcal{L}_i,i=1,2,3$ the fluctuations are given by a product of two $\GOE$-distributions. At the triple point, the fluctuations are given by a product of three $\GOE$ distributions.}
\label{FigShockCollisions}
\end{center}
\end{figure}

If we start from the triple point and take the parameter $\tau\to\pm\infty$ following the directions of the standard shocks (the blue lines in Figure~\ref{FigShockCollisions}), then the fluctuations reduces to the ones of a standard shock, i.e., they are given in our case as product of two GOE Tracy-Widom distributions (appropriately centered and scaled) as stated below.

Physically, when $\tau\ll -1$, we are in the situation where the two normal shocks have not yet met and thus their distribution is given as a product of two $F_{\rm GOE}$ distributions: Recalling the line segments from Figure~\ref{FigShockCollisions}, the situation in (a) below corresponds to look at  $\mathcal{L}_{1}$, (b) below to look at $\mathcal{L}_{2}.$   When $\tau\gg 1$, the two shocks have merged into a single one and thus also in this case the distribution is a product of two $F_{\rm GOE}$. This corresponds to look at the line segment $\mathcal{L}_{3}$ and is done in (c) below.
\begin{cor}\label{CorStandardShocks}
We have the following transitions to standard shock behavior:
\begin{itemize}
\item[(a)] Let us replace $u\to u+\tau\rho_1\rho_2$, $s\to s-(1-\rho_1-\rho_2)\tau$. Then
\begin{equation}
\lim_{\tau\to-\infty} \prod_{k\in\{1,2,3\}}F_{\rm GOE}((s-\mu_k u+\nu_k\tau)/\sigma_k) = \prod_{k\in\{1,2\}} F_{\rm GOE}((s-\mu_k u)/\sigma_k).
\end{equation}
\item[(b)] Let us replace  $u\to u+\tau\rho_2\rho_3$, $s\to s-(1-\rho_2-\rho_3)\tau$. Then
\begin{equation}
\lim_{\tau\to-\infty} \prod_{k\in\{1,2,3\}}F_{\rm GOE}((s-\mu_k u+\nu_k\tau)/\sigma_k) = \prod_{k\in\{2,3\}} F_{\rm GOE}((s-\mu_k u)/\sigma_k).
\end{equation}
\item[(c)] Let us replace  $u\to u+\tau\rho_1\rho_3$, $s\to s-(1-\rho_1-\rho_3)\tau$. Then
\begin{equation}
\lim_{\tau\to\infty} \prod_{k\in\{1,2,3\}}F_{\rm GOE}((s-\mu_k u+\nu_k\tau)/\sigma_k) = \prod_{k\in\{1,3\}} F_{\rm GOE}((s-\mu_k u)/\sigma_k).
\end{equation}
\end{itemize}
\end{cor}
\begin{proof}
The proof consists in elementary computations. For instance, in the first case, after shifting of $u$ and $s$ leads to
\begin{equation}
\begin{aligned}
s-\mu_1 u+\nu_1 \tau &\to s-\mu_1 u,\quad s-\mu_2 u+\nu_2 \tau \to s-\mu_2 u, \\
s-\mu_3 u+\nu_3 \tau &\to s-\mu_3 u-\frac{\tau}{\rho_3}(\rho_3-\rho_2)(\rho_3-\rho_1)\to \infty
\end{aligned}
\end{equation}
as $\tau\to -\infty$ and $F_{\rm GOE}(\infty)=1$.
\end{proof}

\paragraph{Outline.} In Section~\ref{sectFixedDensity} we analyze in details TASEP with fixed density $\rho$ over all $\Z$. We state the one-point distribution, the slow decorrelation and the localization result, which is the main work. These are then used as input in Section~\ref{SectProof} to prove the main theorem.

\paragraph{Acknowledgements.} This work is supported  by the Deutsche Forschungsgemeinschaft (German Research Foundation) by  the CRC 1060  (Projektnummer 211504053) and  Germany's Excellence Strategy - GZ 2047/1, Projekt ID 390685813.

\section{TASEP with fixed density}\label{sectFixedDensity}
In order to prove Theorem~\ref{ThmMain}, we start by studying the localization property of TASEP with constant density everywhere, i.e., for some fixed density $\rho\in (0,1)$, consider the initial condition
\begin{equation}
x_n(0)=-\lfloor n/\rho\rfloor,\quad n\in\Z.
\end{equation}

The characteristic lines have speed $1-2\rho$, while particles in average move with speed $1-\rho$. We want to consider particles which are non-trivially correlated with the event around the origin at time $0$. Thus we get\footnote{We will not write explicitly the integer valued in the following, since they are irrelevant for the asymptotic question we are going to study.}
\begin{equation}
  N=\lfloor \rho^2 t\rfloor,
\end{equation}
and, at time $t$,
\begin{equation}
x_N(t)\simeq (1-2\rho)t.
\end{equation}

The fluctuations of $x_n(t)$ are known to be GOE Tracy-Widom distributed. Indeed, by rewriting the result proven in~\cite{FO17} in terms of distribution of particles instead of height function or last passage percolation we have the following.
\begin{prop}[Theorem 2.1 in ~\cite{FO17}]\label{PropCvgDistr}
Consider TASEP with initial condition $x_n(0)=-\lfloor n/\rho\rfloor$ for some fixed $\rho\in (0,1)$. Then
\begin{equation}
\lim_{t\to\infty} \Pb(x_{\rho^2 t+\rho\alpha t}(t)\geq (1-2\rho-\alpha) t - s t^{1/3})=F_{\rm GOE}(s/\sigma(\rho)),
\end{equation}
with $\sigma(\rho)=2^{2/3}\rho^{1/3}/(1-\rho)^{2/3}$.
\end{prop}

\subsubsection*{The backwards path}
We define the following process running \emph{backwards} in time as in~\cite{Fer18}, see Figure~\ref{FigSpaceCut} for an illustration. Let $N(t)=N$. If at time $\hat t^+$ we have $N(\hat t^+)=n$ and at time $\hat t$ a jump of particle $n$ is suppressed by the presence of particle $n-1$, then we set $N(\hat t)=n-1$. Further, for any $u \in (0,t)$, we define the particle system $\{\hat x_n(s),u\leq s\leq t, n\geq N(u)\}$ by setting the position of particles at $s=u$ as $\hat x_n(u)=x_{N(u)}(u)-n+N(u)$, for $n\geq N(u)$. The evolution of $\hat x_n$'s and $x_n$ are coupled by the basic coupling.
\begin{figure}[t!]
\begin{center}
\psfrag{x}[lc]{$x$}
\psfrag{t}[lb]{$t$}
\psfrag{tau}[lb]{$u$}
\includegraphics[height=5cm]{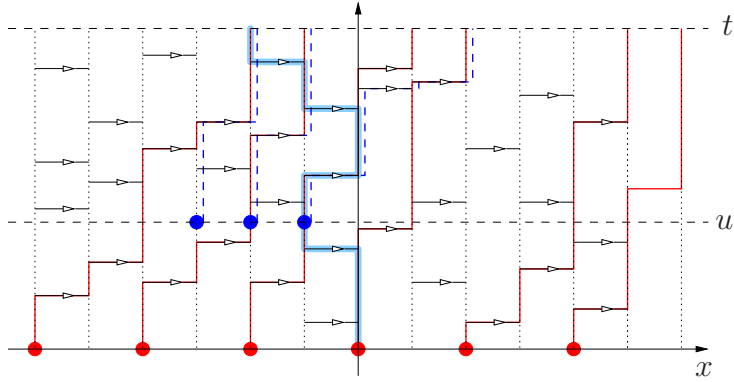}
\caption{The red solid lines are the trajectories of $\{x_n(s),0\leq s\leq t, \textrm{ all }n\}$. The thick light-blue line is the trajectory of $\{x_{N(s)}(s),0\leq s\leq t\}$. The dashed blue lines are the trajectories of $\{\hat x_{n}(s), u\leq s\leq t,n\geq N(u)\}$.}
\label{FigSpaceCut}
\end{center}
\end{figure}
We define for $Z \in \Z$ the step initial data
\begin{equation}
x_{n}^{{\rm step},Z }=-n+Z+1, n\geq 1,
\end{equation}
for $Z=0$ we simply write $x_{n}^{\rm{step} }$.
In Proposition~3.4 of~\cite{Fer18}  it is shown that $x_N(t)=\hat x_N(t)$, i.e.,
\begin{equation}\label{eq2.6}
x_N(t)=x_{N(\tau)}(\tau)+y_{N(\tau)}(\tau,t),\quad y_{N(\tau)}(\tau,t)\stackrel{(d)}{=}x^{\rm step}_{N-N(\tau)+1}(t-\tau),
\end{equation}
Here, the process $y_{N(\tau)}(\tau,t)$ from time $\tau$ to time $t$ is obtained by starting \textit{at time} $\tau$ a TASEP from $\{x_{n}^{{\rm step},Z}\}_{n\geq 1}$ with $Z=x_{N(\tau)}(\tau)$
and setting $ y_{N(\tau)}(\tau,t)=x_{N-N(\tau)+1}^{{\rm step},Z}(t)-Z.$
Furthermore, for any other $k\leq N$ it holds
\begin{equation}\label{eq1.5}
x_N(t)\leq x_{k}(\tau)+y_{k}(\tau,t),\quad y_{k}(\tau,t)\stackrel{(d)}{=}x^{\rm step}_{N-k+1}(t-\tau).
\end{equation}
In other words, we can write
\begin{equation}\label{eqMin}
x_N(t)=\inf_{k\leq N} x^{{\rm step},x_{k}(0)}_{N-k+1}(t) = x^{{\rm step},x_{N(0)}(0)}_{N-N(0)+1}(t).
\end{equation}
This was also previously shown in Lemma~2.1 of~\cite{Sep98c}.

\begin{remark}\label{remarkLoc}
An important property of the construction of the backwards paths is the following (see~\cite{Fer18} for more detailed explanations). Suppose that the backwards path $x_{N(u)}(u)$, $u\in[0,t]$, stays inside a deterministic region $\cal C$. Then, the distribution of $x_{N(t)}(t)$ does not depend on the particular realization of the randomness outside $\cal C$. Thus, if the backwards path stays inside $\cal C$ with probability going to one as $t\to\infty$, then the law of $x_{N(t)}(t)$ is asymptotically independent of the randomness outside $\cal C$.
\end{remark}

From a law of large number and the KPZ scaling exponents, we expect
\begin{equation}\label{eq2.9}
  N(\tau)=\rho^2 \tau + \Or(t^{2/3}),\quad x_{N(\tau)}=(1-2\rho)\tau + \Or(t^{2/3})
\end{equation}
as we shall prove below.

\subsubsection*{Localization}
Consider the cylinder $\cal C$ around the characteristic line leaving from the origin of width $t^{2/3+\e}$, namely
\begin{equation}
{\cal C}=\{(x,u)\in \R\times [0,t]\, |\, |x-(1-2\rho)u|< t^{2/3+\e}\}.
\end{equation}
Define its left/right border as
\begin{equation}
{\cal L}_\pm = \{(x,u)\in \R\times [0,t]\, |\, x=(1-2\rho)u\pm t^{2/3+\e}\}.
\end{equation}
Consider the events
\begin{equation}
\Omega_\e=\{\exists u\in [0,t] \, | \, (x_{N(u)}(u),u)\not\in{\cal C}\},\quad G=\{(x_N(t),t)\in {\cal C}\}.
\end{equation}

\begin{thm}[Localization]\label{thmLocalization}
Let $0<\e<1/3$. Then, for all $t$ large enough,
\begin{equation}
\Pb(\Omega_\e)\leq C e^{-c t^\e},
\end{equation}
for some constants $C,c>0$.
\end{thm}

To prove Theorem~\ref{thmLocalization}, we need to have a good control on the fluctuations of $x_N(t)$.
\begin{prop}\label{PropPtLineBounds}
Let $\alpha\in\R$ be fixed and consider $N=\rho^2 t+\rho\alpha t$. Let \mbox{$\delta\in (0,1/3)$}. Then, for all $t$ large enough,
\begin{equation}
\begin{aligned}
\Pb(x_N(t)\geq (1-2\rho-\alpha) t+t^{1/3+\delta}) &\leq C e^{-c t^{3\delta/2}},\\
\Pb(x_N(t)\leq (1-2\rho-\alpha) t-t^{1/3+\delta}) &\leq \widetilde C e^{-\tilde c t^{\delta}},
\end{aligned}
\end{equation}
for some constants $C,\widetilde C,c,\tilde c$.
\end{prop}
\begin{proof}
The system is basically translation invariant (up to order one) and thus we can w.l.o.g.\ consider $\alpha=0$. Since by basic coupling $x_N(t)\leq x^{\rm step}_N(t)$,we have
\begin{equation}
\Pb(x_N(t)\geq (1-2\rho) t+t^{1/3+\delta}) \leq \Pb(x^{\rm step}_N(t)\geq (1-2\rho) t-st^{1/3})
\end{equation}
with $s=-t^\delta$. Using \eqref{BoundsPtPt} we thus get the claimed result.

Union bound applied to \eqref{eqMin} gives
\begin{equation}
\begin{aligned}
&\Pb(x_N(t)\leq (1-2\rho) t-t^{1/3+\delta}) \\
&\qquad\leq \sum_{N_0\leq N} \Pb(x^{\rm step}_{N-N_0+1}(t)\leq (1-2\rho) t+N_0/\rho-t^{1/3+\delta})\\
&\qquad\qquad=\sum_{\eta\in t^{-1}\Z_{\leq t\rho^2}} \Pb(x^{\rm step}_{t(\rho^2-\eta)}(t)\leq (1-2\rho) t+\eta t/\rho-t^{1/3+\delta}).
\end{aligned}
\end{equation}

We want to use the bounds of Lemma~\ref{lemBound2}. These bounds holds uniformly for particle index $\nu t$ with $\nu$ is a bounded set of $(0,1)$. Thus, if particle number is close to $0$ or $t$, we need to bound their distributions by the ones of other particles with particle number macroscopically away from $0$ and $t$.

\emph{Case 1:} Consider $\eta\in [\tfrac34 \rho^2,\rho^2-1/t]$. Then, since particle with index $t(\rho^2-\eta)=N-N_0$ is to the right of particle with index $N-\tfrac34\rho^2 t$, we get
\begin{multline}
\Pb(x^{\rm step}_{t(\rho^2-\eta)}(t)\leq (1-2\rho) t+\eta t/\rho-t^{1/3+\delta}) \leq
\Pb(x^{\rm step}_{t(\rho^2-\eta)}(t)\leq (1-\rho) t-t^{1/3+\delta})\\
\leq \Pb(x^{\rm step}_{\rho^2 t/4}(t)\leq (1-\rho) t-t^{1/3+\delta})\leq C_2 e^{-c_2 t^\delta},
\end{multline}
where we use \eqref{BoundsPtPt}.

\emph{Case 2:} Consider $\eta<\rho^2-\rho$ and set $\nu=\rho^2-\eta$, which is strictly larger than $\rho$. Then, $x_{N-N_0}^{\rm step}(t)\geq x_{N-N_0}^{\rm step}(0)=-\nu t$. Thus
\begin{equation}
\Pb(x^{\rm step}_{t(\rho^2-\eta)}(t)\leq (1-2\rho) t+\eta t/\rho-t^{1/3+\delta}) = 0
\end{equation}
whenever $-\nu t > (1-2\rho)t+\eta t/\rho-t^{1/3+\delta}$. This holds true whenever \mbox{$\nu>\rho-\frac{\rho}{1-\rho} t^{-2/3+\delta}$}. Thus for $\eta<\rho^2-\rho$ and any large enough $t$, the searched probability is strictly equal to zero.

\emph{Case 3:} Finally consider the case $\eta\in [\rho^2-\rho,\tfrac34\rho^2]$ and let $\nu=\rho^2-\eta$.
We have bounds on the fluctuations of $x^{\rm step}_{\nu t}(t)$ with respect to $(1-2\sqrt{\nu})t$. Solving
\begin{equation}
(1-2\sqrt{\nu})t-st^{1/3}=(1-2\rho+\eta/\rho)t-t^{1/3+\delta}
\end{equation}
we get
\begin{equation}\label{eq1.14}
s=2\rho\Big(1-\frac{\eta}{2\rho^2}-\sqrt{1-\frac{\eta}{\rho^2}}\Big)t^{2/3}+t^{\delta}.
\end{equation}
For $\eta\in [-\rho^2,\rho^2]$, $\eqref{eq1.14}\geq t^{\delta}+\eta^2 t^{2/3}/(8\rho^3)$, while for $\eta<-\rho^2$ we have
$\eqref{eq1.14}\geq t^{\delta}+\eta t^{2/3}/(8\rho)$.

Thus, using \eqref{BoundsPtPt}, we get
\begin{equation}
\Pb(x^{\rm step}_{t(\rho^2-\eta)}(t)\leq (1-2\rho) t+\eta t/\rho-t^{1/3+\delta})\leq  C_2 e^{-c_2 t^{\delta}} e^{-c_2 \frac{t^{2/3}}{8\rho^3}\min\{\eta^2,\eta\rho^2\}}.
\end{equation}

Combining the bounds for the three cases and performing the sum over $\eta\in t^{-1} \Z$ leads to the desired result.
\end{proof}

As a consequence, if we know that $x_{N(\tau)}(\tau)=(1-2\rho)\tau - t^{2/3+\e}$, then we have a good estimate on $N(\tau)$, namely \eqref{eq2.9}, as stated in Lemma~\ref{lem2.4} and Lemma~\ref{lem2.5} below. With this we can estimate $x^{\rm step}_{N-N(\tau)+1}(t-\tau)$ and conclude that $x_{N(\tau)}(\tau)+y_{N(\tau)}(\tau,t)$ (which is equal to  $x_N(t)$) is much larger than its typical value. This gives us that it is actually very unlikely to have  $x_{N(\tau)}(\tau)=(1-2\rho)\tau - t^{2/3+\e}$.

\begin{lem}\label{lem2.4}
Let $0<\e<1/3$. Consider  the event that at time $\tau t$ a particle is at position $(1-2\rho)\tau t-t^{2/3+\e}$, i.e.,
\begin{equation}
E_-=\{x_{n}(\tau t)=(1-2\rho)\tau t-t^{2/3+\e}\textrm{ for some }n\in \Z\}.
\end{equation}
On the event $E_-$, let $N_{0}^{-}$ be s.t. $x_{N_{0}^{-}}(\tau t)=(1-2\rho)\tau t-t^{2/3+\e}$.
Then, for all $t$ large enough,
\begin{equation}
\Pb(\{|N_{0}^{-}-\rho^2\tau t -\rho t^{2/3+\e}|\geq \rho t^{1/3+\e}\}\cap E_-)\leq C e^{-c t^{\e}},
\end{equation}
for some constants $C,c>0$ independent of $\tau \in [0,1]$.
\end{lem}
\begin{proof}
We prove that the fluctuations of $N_{0}^{-}$ around $\rho^2\tau t+\rho t^{2/3+\e}$ are not of larger order than $t^{1/3}$.

\emph{Case 1: $N_{0}^{-}<\rho^2\tau t+\rho t^{2/3+\e}-\rho t^{1/3+\e}=:(\rho^2+\alpha\rho)\tau t$}, so that \mbox{$\alpha=\tau^{-1} (t^{-1/3+\e}-t^{-2/3+\e})$}. Define the event
${\cal N}_-=\{N_{0}^{-}<(\rho^2+\alpha\rho)\tau t\}$. Then $x_{N_{0}^{-}}(\tau t)\geq x_{(\rho^2+\alpha\rho)\tau t}(\tau t)$
so that
\begin{equation}\label{eq1.18}
\begin{aligned}
\Pb({\cal N}_-\cap E_-)&\leq \Pb({\cal N}_-\cap \{x_{N_{0}^{-}}(\tau t)\leq (1-2\rho) \tau t-t^{2/3+\e}\})\\
&\leq \Pb({\cal N}_-\cap \{x_{(\rho^2+\alpha\rho)\tau t}(\tau t)\leq (1-2\rho)\tau t-t^{2/3+\e}\})\\
&\leq \Pb(x_{(\rho^2+\alpha\rho)\tau t}(\tau t)\leq (1-2\rho)\tau t-t^{2/3+\e})\\
&=\Pb(x_{(\rho^2+\alpha\rho)\tau t}(\tau t)\leq (1-2\rho-\alpha)\tau t-t^{1/3+\e})\\
&\leq \Pb(x_{(\rho^2+\alpha\rho)\tau t}(\tau t)\leq (1-2\rho-\alpha)\tau t-(\tau t)^{1/3+\e})\\
&\leq   \widetilde C e^{-\tilde c t^{\e}},
\end{aligned}
\end{equation}
where in the last step we used the second inequality of Proposition~\ref{PropPtLineBounds} with $t\to \tau t$.

\emph{Case 2: $N_{0}^{-}>\rho^2\tau t+\rho t^{2/3+\e}+\rho t^{1/3+\e}$}. This case is treated similarly, with just $\leq$ inequalities in the events replaced by $\geq$, and in the last step we use the first (stronger) inequality of Proposition~\ref{PropPtLineBounds}.
By renaming the constants we get the claimed result.
\end{proof}

In exactly the same way one proves the following.
\begin{lem}\label{lem2.5}
Let $0<\e<1/3$. Consider for the event that at time $\tau t$ a particle is at position $(1-2\rho)\tau t+t^{2/3+\e}$, i.e.,
\begin{equation}
E_+=\{x_{n}(\tau t)=(1-2\rho)\tau t+t^{2/3+\e}\textrm{ for some }n\in \Z\}.
\end{equation}
On the event $E_+$, let $N_{0}^{+}$ be s.t. $x_{N_{0}^{+}}(\tau t)=(1-2\rho)\tau t+t^{2/3+\e}$.
Then, for all $t$ large enough,
\begin{equation}
\Pb(\{|N_{0}^{+}-\rho^2\tau t +\rho t^{2/3+\e}|\geq \rho t^{1/3+\e}\}\cap E_+)\leq C e^{-c t^{\e}},
\end{equation}
for some constants $C,c>0$ independent of $\tau \in [0,1]$.
\end{lem}

Now that we have a good control of the particle number, we can get good estimates on $x^{\rm step}_{N-N_{0}^{\pm}+1}((1-\tau)t)$.
\begin{prop}\label{propBoundStep}
Let $E_\pm, N_{0}^{\pm}$ be  defined as above. Then
\begin{equation}\label{eq2.26}
\Pb(\{x^{\rm step}_{N-N_0^{-}+1}((1-\tau) t)\leq (1-2\rho)(1-\tau)t+t^{2/3+\e}+\tfrac18 t^{1/3+2\e}\}\cap E_-)\leq C e^{-c t^\e},
\end{equation}
and
\begin{equation}\label{eq2.27}
\Pb(\{x^{\rm step}_{N-N_0^{+}+1}((1-\tau) t)\leq (1-2\rho)(1-\tau)t-t^{2/3+\e}+\tfrac18 t^{1/3+2\e}\}\cap E_+)\leq C e^{-c t^\e}.
\end{equation}
\end{prop}
\begin{proof}We give the details of the proof only for \eqref{eq2.26}, since the proof of \eqref{eq2.27} is analogous. Let  ${\cal N}_-=\{|N_{0}^{-}-\rho^2\tau t-\rho t^{2/3+\e}|>\rho t^{1/3+\e}\}$ and \mbox{$P_-=\{x^{\rm step}_{N-N_{0}^{-}+1}((1-\tau) t)\leq (1-2\rho)(1-\tau)t+t^{2/3+\e}+\tfrac18 t^{1/3+2\e}\}$}.
Then
\begin{equation}
\begin{aligned}
\Pb(P_-\cap E_-)&\leq \Pb(P_-\cap E_- \cap {\cal N}_-^c)+\Pb(E_-\cap {\cal N}_-)\\
&\leq \Pb(P_-\cap {\cal N}_-^c)+ C e^{-ct^\e}
\end{aligned}
\end{equation}
where in the second step we used Lemma~\ref{lem2.4}. In the event ${\cal N}_-^c$, we know that $N_{0}^{-}=\rho^2\tau t+\rho t^{2/3+\e}+\gamma t^{1/3+\e}$ for some (random) $|\gamma|<1$. Setting $N-N_{0}^{-}=\nu (1-\tau)t$ an easy computation gives
\begin{multline}
(1-2\rho)(1-\tau)t+t^{2/3+\e}+\tfrac18 t^{1/3+2\e}-(1-2\sqrt{\nu})(1-\tau)t \\
=\Big(\frac18-\frac{1}{2\rho(1-\tau)}\Big) t^{1/3+2\e}+\Or(t^{1/3+\e}; t^{3\e})\leq -\tfrac18 t^{1/3+2\e}
\end{multline}
for all $t$ large enough. Thus, using Lemma~\ref{lemBound2}, we get
\begin{equation}
\Pb(P_-\cap {\cal N}_-^c) \leq\Pb(x^{\rm step}_{\nu(1-\tau)t}((1-\tau)t)\leq (1-2\sqrt{\nu})(1-\tau)t-\tfrac18 t^{1/3+2\e})\leq 2 C e^{-c t^\e},
\end{equation}
for all $t$ large enough, which finishes the proof.
\end{proof}

Now we can prove the localization theorem.
\begin{proof}[Proof of Theorem~\ref{thmLocalization}]
If $\omega\in\Omega_\e\cap G$, then the backwards path starts inside the cylinder and at some time crosses its boundary ${\cal L}_+\cup{\cal L}_-$. Denote by ${\cal I}_\pm(t_0)$ the events that there is a crossing in the time interval $[t_0,t_0+1)$, i.e. \mbox{${\cal I}_\pm(t_0)=\{x_{N(t_1)}(t_1)\in{\cal L}_\pm\textrm{ for some }t_1\in [t_0,t_0+1)\}$}. Then
\begin{equation}\label{eq2.32}
\Omega_\e\cap G = \bigcup_{t_0 \in \Z \cap [0,t]}{\cal I}_-(t_0)\cup{\cal I}_+(t_0).
\end{equation}

First let us control that the displacement in a time interval bounded by $1$ is also small. For any backwards path, define the event that in a time distance bounded by $1$ the increment is bounded by $t^\e$, namely for $t_0\in \Z\cap[0,t]$, set
\begin{equation}
\begin{aligned}
D_{t_0}=&\{|x_{N(t_0)}(t_0)-x_{N(t_1)}(t_1)|\leq t^{\e}\textrm{ for all }t_1\in [t_0,t_0+1)\}\\
&\cap\{|y_{N(t_0)}(t_0,t)-y_{N(t_1)}(t_1,t)|\leq t^{\e}\textrm{ for all }t_1\in [t_0,t_0+1)\},
\end{aligned}
\end{equation}
and let
\begin{equation}
D=\bigcap_{t_0\in\Z\cap[0,t]}D_{t_0}.
\end{equation}
By the graphical construction it is immediate that $|x_{N(t_0)}(t_0)-x_{N(t_1)}(t_1)|$ is dominated by a Poisson process with intensity $1$. Furthermore, since \mbox{$x_N(t)=x_{N(t_0)}(t_0)+y_{N(t_0)}(t_0,t) = x_{N(t_1)}(t_1)+y_{N(t_1)}(t_1,t)$}, the two events in $D_{t_0}$ are actually the same. Therefore we have
\begin{equation}
\Pb(D^c)\leq \sum_{t_0\in\Z\cap[0,t]}\Pb(D_{t_0}^c)\leq C te^{-2 t^{\e}}\leq C e^{-t^\e}.
\end{equation}

Define the event
\begin{equation}
F=\{x_N(t)\leq (1-2\rho)t+\tfrac18 t^{1/3+2\e}-2t^\e\}.
\end{equation}
Then, first of all,
\begin{equation}
\Pb(\Omega_\e)\leq \Pb(\Omega_\e\cap G\cap D\cap F)+\Pb(F^c) + \Pb(G^c) + \Pb(D^c).
\end{equation}
By Proposition~\ref{PropPtLineBounds} with $\alpha=0$ we know that $\Pb(G^c)\leq C e^{-c t^{\e}}$ and also $\Pb(F^c)\leq  C e^{-c t^\e}$. Thus it remains to get a bound for $\Pb(\Omega_\e\cap G\cap D\cap F)$.

\emph{Case 1: For some $t_0$ integer, ${\cal I}_-(t_0)$ and $D$ hold.} In this case, we have $x_N(t)\geq (1-2\rho)t_0-t^{2/3+\e}+x^{\rm step}_{N-N(t_0)+1}(t-t_0)-1-t^\e$. By Proposition~\ref{propBoundStep}, we get then
\begin{equation}\label{eq2.36}
x_N(t)\geq (1-2\rho)t+\tfrac18 t^{1/3+2\e}-1-t^\e
\end{equation}
on a set of probability bounded from below by $1-C e^{-c t^\e}$. But since the condition \eqref{eq2.36} and the one defining $F$ are incompatible, we have
\begin{equation}
\Pb({\cal I}_-(t_0)\cap D\cap F)\leq C e^{-c t^\e}.
\end{equation}

\emph{Case 2: For some $t_0$ integer, ${\cal I}_+(t_0)$ and $D$ hold.} In this case, we have $x_N(t)\geq(1-2\rho)t_0+t^{2/3+\e}+x^{\rm step}_{N-N_0+1}(t-t_0)-1-t^\e$ and by using Proposition~\ref{propBoundStep} we get as well
\begin{equation}
\Pb({\cal I}_+(t_0)\cap D\cap F)\leq C e^{-c t^\e}.
\end{equation}

Thus, we have shown that
\begin{equation}
\Pb(\Omega_\e\cap G\cap D\cap F) \leq \sum_{t_0\in\Z\cap[0,t]}\sum_{\sigma\in\{-,+\}} \Pb({\cal I}_\sigma(t_0)\cap D\cap F) \leq 2 C te^{-c t^{\e}}.
\end{equation}
By redefining the constants the result claimed is proven.\end{proof}

\subsubsection*{Slow decorrelation}
For KPZ models, typically (away from shocks for instance), the fluctuation of two particles at time $t$ are non-trivially correlated if their distance is of order $t^{2/3}$, i.e., their relative fluctuations is of the same order of the fluctuation of a single one, namely in the $t^{1/3}$ scale.

The last ingredient of the proof of the main result is the following statement. This is called slow decorrelation because it says that the fluctuations of a particle at time $t$ differs from the fluctuations of a particle at time $t-o(t)$ only by an amount which of order $o(t^{1/3})$, provided that the two particles are on a given characteristic line, i.e., the decorrelation along characteristic lines occurs over macroscopic time scales, in contrast to the $t^{2/3}$ scale for spatial correlations.
\begin{thm}[Slow decorrelation]\label{thmSlowDec}
For all $\delta>0$ and some $0<\nu<1$,
\begin{equation}\label{eqSlowDec}
\lim_{t\to\infty} \Pb(|x_{\rho^2 t}(t)-x_{\rho^2 (t-t^\nu)}(t-t^\nu)-(1-2\rho)t^\nu|\geq \delta t^{1/3})=0.
\end{equation}
\end{thm}
\begin{proof}
Slow decorrelation was proven in~\cite{CFP10b} in terms of height functions. In Proposition~3.3 of~\cite{Fer18} the proof have been adapted to the observable given by the particle position. The key input is the existence of non-trivial distributions $D,D'$ such that
\begin{equation}\label{eq2.42}
\begin{aligned}
&\lim_{t\to\infty}\frac{x_{\rho^2 t}(t)-(1-2\rho)t}{-t^{1/3}}\stackrel{(d)}{=}D,\\
&\lim_{t\to\infty}\frac{x_{\rho^2 (t-t^\nu)}((t-t^\nu))-(1-2\rho)(t-t^\nu)}{-t^{1/3}}\stackrel{(d)}{=}D,\\
&\lim_{t\to\infty}\frac{x^{\rm step}_{\rho^2 t^\nu}(t^\nu)-(1-2\rho)t^\nu}{-t^{\nu/3}}\stackrel{(d)}{=}D'.
\end{aligned}
\end{equation}
Then one uses \eqref{eq1.5}, namely $x_{\rho^2 t}(t)\leq x_{\rho^2 (t-t^\nu)}(t-t^\nu) + x^{\rm step}_{\rho^2 t^\nu}(t^\nu)$, to deduce \eqref{eqSlowDec}. The distribution function $D$ is given in terms of the GOE Tracy-Widom distribution function in~\cite{FO17}, see Proposition~\ref{PropCvgDistr}, while the convergence of the step initial condition is well known to be the GUE Tracy-Widom distribution function~\cite{Jo00b}.
\end{proof}

\section{Proof of Theorem~\ref{ThmMain}}\label{SectProof}
The strategy of the proof is the following (see Figure~\ref{FigDecompProof} for an illustration):
\begin{enumerate}
\item First we decompose the problem into three simpler initial conditions, so that $x_N(t)=\min\{x^1_N(t),x^2_N(t),x^3_N(t)\}$, see Proposition~\ref{propMinimum}.
\item For each  of the simpler problems, we can construct the backwards paths and show that, for all $\nu<1$, there are constants $c_k,N_k$ such that $x^k_{N}(t)-x^k_{N_k}(t-T^\nu)-c_k T^\nu=o(T^{1/3})$, i.e., the fluctuations at time $t$ are governed by the ones at time $t-T^\nu$ already, see Corollary~\ref{corSlowDec}.
\item Finally we show that the backwards paths for $x^k_{N_k}(t-T^\nu)$ are localized inside cylinders of width $t^{2/3+\e}$ with high probability, see Corollary~\ref{corLocalization}. This implies that $x^k_{N_k}(t-T^\nu)$, $k=1,2,3$ are asymptotically independent, as the cylinders, for all $T$ large enough, are disjoint provided that $\nu>2/3+\e$. From this asymptotic independence, we  obtain the product form of the limiting distribution of $x_N(t)$.
\end{enumerate}

\begin{figure}
\begin{center}
\psfrag{C1}[r]{${\cal C}_1$}
\psfrag{C2}[c]{${\cal C}_2$}
\psfrag{C3}[l]{${\cal C}_3$}
\psfrag{t}[lb]{$t$}
\psfrag{t1}[lb]{$t-T^\nu$}
\includegraphics[height=5cm]{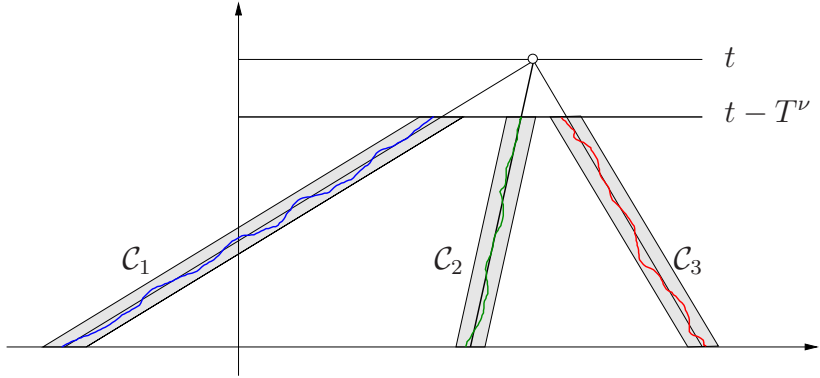}
\caption{By slow-decorrelation, the fluctuations of $x^k_N(t)$ are essentially the same as $x^k_{N_k}(t-T^\nu)$ with $N_k=N-\rho_k T^\nu$. The colored irregular lines represent the backwards paths starting at $x^k_{N_k}(t-T^\nu)$, $k=1,2,3$, whose laws depend only on randomness in the cylinder ${\cal C}_k$ with high probability (localization). These cylinder are however disjoint for $2/3+\e<\nu<1$, giving asymptotic independence.}
\label{FigDecompProof}
\end{center}
\end{figure}

\paragraph{I. Decomposition into three simpler problems.}
The distribution of $x_n(t)$ can be written as the minimum of three different initial conditions, see Figure~\ref{FigDecomposition}. These are obtained as follows:
\begin{itemize}
\item[(1)] all particles to the right of the origin are removed,
\begin{equation}
x_n^1(0)=-\lfloor n/\rho_1\rfloor, \textrm{for }n\geq 0,
\end{equation}
\item[(2)] all particles to the right of $T$ are removed, while particles to the left of the origin are packed into a step initial condition to the left of $0$,
\begin{equation}
x_n^2(0)=\left\{
\begin{array}{ll}
-n,& \textrm{for }n\geq 0,\\
-\lfloor n/\rho_2\rfloor,& \textrm{for }-1\leq n\leq -\lfloor T\rho_2\rfloor,
\end{array}
\right.
\end{equation}
\item[(3)] all particles to the left of $T$ are packed into a step initial condition to the left of $T$.
\begin{equation}
x_n^3(0)=\left\{
\begin{array}{ll}
T-(n+\lfloor T\rho_2\rfloor),& \textrm{for }n\geq -\lfloor T\rho_2\rfloor,\\
T-\lfloor (n+\lfloor T\rho_2\rfloor)/\rho_3\rfloor,& \textrm{for } n<-\lfloor T\rho_2\rfloor.
\end{array}
\right.
\end{equation}
\end{itemize}
\begin{figure}
\begin{center}
\psfrag{r1}[cb]{$\rho_1$}
\psfrag{r2}[cb]{$\rho_2$}
\psfrag{r3}[ct]{$\rho_3$}
\psfrag{0}[cb]{$0$}
\psfrag{T}[cb]{$T$}
\includegraphics[height=3.5cm]{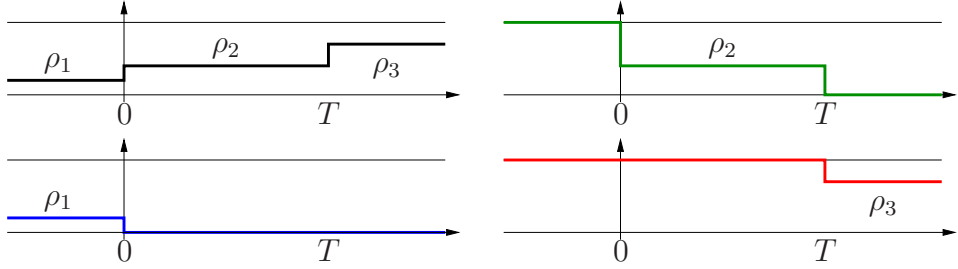}
\caption{Densities of the different initial conditions. The original model is illustrated in the top-left figure, $x^1(0)$ in the bottom-left, $x^2(0)$ in the top-right, and $x^3(0)$ in the bottom-right figure.}
\label{FigDecomposition}
\end{center}
\end{figure}

Using the graphical construction for TASEP~\cite{Lig76,Har78}, see also~\cite{Li99}, one can couple the processes $\{x_n(t),n\in\Z\}$, $\{x^1_n(t),n\geq 0\}$, $\{x^2_n(t),n\geq -\lfloor T\rho_2\rfloor\}$, and $\{x^3_n(t),n\in\Z\}$. With this basic coupling, since TASEP is attractive, one immediately has, for any $n\geq 1$,
\begin{equation}
x_n(t)\leq x_n^1(t),\quad x_n(t)\leq x_n^2(t),\quad x_n(t)\leq x_n^3(t).
\end{equation}
However we have more than these three inequalities.
\begin{prop}[Decomposition in subproblems]\label{propMinimum}
Let us consider $N>0$. Then we have the following decomposition
\begin{equation}
x_N(t)=\min\{x^1_N(t),x^2_N(t),x^3_N(t)\}.
\end{equation}
\end{prop}
\begin{proof}
There are two ways of proving it. One is a slight extension of the one of Lemma~3.1 of~\cite{Fer18}. The idea is to see that if the index of the backwards path at time $0$ satisfies $N(0)\geq 0$, then $x_N(t)=x_N^1(t)$, while if \mbox{$-\lfloor \rho_2 T\rfloor \leq N(0)<0$}, then $x_N(t)=x_N^2(t)$, and finally if $N(0)<-\lfloor \rho_2 T\rfloor$, then $x_N(t)=x_N^3(t)$. Alternatively, we can apply the representation of Lemma~2.1 of~\cite{Sep98c}, see \eqref{eqMin}, and notice that
\begin{equation}
\begin{aligned}
x_N^1(t)&=\min_{N\geq k\geq 0}x^{{\rm step},x_{k}^{1}(0)}_{N-k+1}(t),\\
x_N^2(t)&=\min_{0\geq k\geq -\lfloor \rho_2 T\rfloor} x^{{\rm step}, x_{k}^{2}(0)}_{N-k+1}(t),\\
x_N^3(t)&=\min_{-\lfloor \rho_2 T\rfloor\geq k} x^{{\rm step},x_{k}^{3}(0)}_{N-k+1}(t).
\end{aligned}
\end{equation}
\end{proof}

\paragraph{II. Extension of the three problems into full-line problems.}
First we extend the models (1), (2) and (3) to have fixed density over all space. Thus define, for all $n\in\Z$,
\begin{equation}
\tilde x_n^1(0)=-\lfloor n/\rho_1\rfloor,\qquad
\tilde x_n^2(0)=-\lfloor n/\rho_2\rfloor,\qquad
\tilde x_n^3(0)=-\lfloor (n+n_0)/\rho_3\rfloor,
\end{equation}
where $n_0=(\rho_3-\rho_2)T$. The shift in $\tilde x_n^3(0)$ is  due to the fact that we need to keep $\tilde x_n^3(0)=x_n^3(0)$ for $n=-\lfloor \rho_2 T\rfloor$.

Let us state a result on the limiting distribution for the above system, which is a corollary of Proposition~\ref{PropCvgDistr}.
\begin{cor}\label{CorFlatResults}
Let us consider
\begin{equation}
N=\frac{\rho_1\rho_2 T}{\rho_3-\rho_1}+u T^{1/3},\quad t=\frac{T}{\rho_3-\rho_1}+\tau T^{1/3},\quad X=\frac{(1-\rho_1-\rho_2)T}{\rho_3-\rho_1}.
\end{equation}
Then
\begin{equation}
\lim_{T\to\infty} \Pb(\tilde x^k_N(t)\geq X-sT^{1/3}) = F_{\rm GOE}((s-\mu_k u+\nu_k\tau)/\sigma_k), \quad k=1,2,3,
\end{equation}
where
\begin{equation}
\sigma_k = \frac{1}{(\rho_3-\rho_1)^{1/3}}\frac{2^{2/3}\rho_k^{1/3}}{(1-\rho_k)^{2/3}},\quad \mu_k = 1/\rho_k,\quad \nu_k=1-\rho_k.
\end{equation}
\end{cor}

\paragraph{III. From $\tilde x^k_N(t)$ to $x^k_N(t)$.}
Recall from \eqref{eqMin} that we have
\begin{equation}
\tilde x^1_{N}(t)=x^{{\rm step}, \tilde x^1_{N(0)}(0)}_{N-N(0)+1}(t),
\end{equation}
where $x^{{\rm step}, \tilde x^1_{N(0)}(0)}_{N-N(0)+1}(t)$ is the position at time $t$ of the $N-N(0)+1$th particle in step-initial condition starting at the position of the starting point of the backwards path, which is $\tilde x^1_{N(0)}(0)$.

Define $\alpha$ such that $N=\rho_1^2 t+\rho_1\alpha t$. As a consequence of the localization result, Theorem~\ref{thmLocalization}, if we take $t$ and $N$ as in Corollary~\ref{CorFlatResults}, then we get
\begin{equation}\label{eq3.11}
\Pb(|\tilde x^1_{N(0)}(0)-\tfrac{\rho_1-\rho_2}{\rho_3-\rho_1} T|\geq T^{2/3+\e})\geq 1-C e^{-c T^\e}.
\end{equation}
Similarly, for the other initial conditions, we have
\begin{equation}\label{eq3.12}
\begin{aligned}
\Pb(|\tilde x^2_{N(0)}(0)-\tfrac{\rho_2-\rho_1}{\rho_3-\rho_1}T)|\geq T^{2/3+\e})&\geq 1-C e^{-c T^\e},\\
\Pb(|\tilde x^3_{N(0)}(0)-\tfrac{2\rho_3-\rho_1-\rho_2}{\rho_3-\rho_1}T|\geq T^{2/3+\e})&\geq 1-C e^{-c T^\e},
\end{aligned}
\end{equation}
for some constants $C,c>0$.

As explained in more details in~\cite{Fer18}, by the construction of the backwards path $x_{N(u)}(u)$, $u\in[0,t]$, at any time we can move away particles to the right of it and move in the right direction particles to its left, without changing the distribution of $x_{N(t)}(t)$. We employed it by moving particles to the right as much as possible and thus generating a step-initial condition. However, if the particle that we moved to create step-initial conditions are moved by a smaller distance, the result is unchanged. Indeed, by monotonicity, the distribution of particle $N$ at time $t$ starting with this new configuration at time $u$, is between $x_{N(t)}(t)$ and $x_{N(u)}(u)+y_{N(u)}(u,t)$, which are however equal by \eqref{eq2.6}.

In particular, on the events that $\tilde x^1_{N(0)}(0)\leq 0$, $\tilde x^2_{N(0)}(0)\in [1,T]$ and $\tilde x^3_{N(0)}(0)>T$, we can modify the initial configurations to $x^1(0)$, $x^2(0)$ and $x^3(0)$ and still have the property $x^k_N(t)=\tilde x^k_N(t)$, $k=1,2,3$. By \eqref{eq3.11} and \eqref{eq3.12}, this happens with probability at least $1-3C e^{-c t^\e}$. Consequently the following corollary holds.
\begin{cor}\label{CorCvgOnePtDistr}
 Corollary~\ref{CorFlatResults} holds with $\tilde x^k_n(t)$ replaced by $x^k_n(t)$ for all $k=1,2,3$.
\end{cor}

\paragraph{IV. Use of slow-decorrelation.}
The reduction of $\tilde x^k_N(t)$ to $x^k_N(t)$ implies also that the slow-decorrelation result, Theorem~\ref{thmSlowDec}, holds here as well. To be explicit we have
\begin{cor}\label{corSlowDec}
Let $N$ and $t$ as in Corollary~\ref{CorFlatResults}. For all $\delta>0$ and some $0<\nu<1$,
\begin{equation}\label{eqSlowDecB}
\lim_{T\to\infty} \Pb(|x^k_{N}(t)-x^k_{N-\rho_k^2 T^\nu}(t-T^\nu)-(1-2\rho_k)T^\nu|\geq \delta T^{1/3})=0,
\end{equation}
for $k=1,2,3$.
\end{cor}
\begin{proof}
As input one needs the convergence in distribution of $x^k_N(t)$ and $x^k_{N_k}(T-T^\nu)$. In the same way as for Corollary~\ref{CorCvgOnePtDistr} for $x_N^k(t)$, one gets for instance
\begin{equation}
\begin{aligned}
&\lim_{T\to\infty}\frac{x^1_{N-\rho_1^2 T^\nu}(t-T^\nu)-\left(X+\nu_1\tau T^{1/3}-\mu_1 u T^{1/3}+(1-2\rho_1)T^\nu\right)}{-T^{1/3}\sigma_1}\stackrel{(d)}{=}F_{\rm GOE},\\
&\lim_{T\to\infty}\frac{x^1_{N}(t)-\left(X+\nu_1\tau T^{1/3}-\mu_1u T^{1/3}\right)}{-T^{1/3}\sigma_1}\stackrel{(d)}{=}F_{\rm GOE},\\
&\lim_{T\to\infty}\frac{x^{\rm step}_{\rho_1^2 T^\nu}(T^\nu)-(1-2\rho_1)T^\nu}{-2^{-1/3}T^{\nu/3}}\stackrel{(d)}{=}F_{\rm GUE},
\end{aligned}
\end{equation}
and similarly for $x^2$ and $x^3$.
\end{proof}

By Proposition~\ref{propMinimum} we have
\begin{equation}\label{eq3.14}
\lim_{T\to\infty} \Pb(x_N(t)\geq X-s T^{1/3}) =\lim_{T\to\infty}\Pb(\min\{x^1_N(t),x_N^2(t),x_N^3(t)\}\geq X-s T^{1/3}).
\end{equation}
Denote by $z^k=x^k_{N-\rho_k^2 T^\nu}(t-T^\nu)+(1-2\rho_k)T^\nu$ for $k=1,2,3$. Then, the slow decorrelation result, Corollary~\ref{corSlowDec}, gives us
\begin{equation}\label{eq3.15}
\eqref{eq3.14} = \lim_{T\to\infty}\Pb(\min\{z^1,z^2,z^3\}\geq X-s T^{1/3}).
\end{equation}

\paragraph{V. Asymptotic independence.}
The remaining step to complete the proof of Theorem~\ref{ThmMain} is to prove that $z^1,z^2,z^3$ are asymptotically independent, see also Remark~\ref{remarkLoc}. For that purpose, define the cylinders
\begin{equation}
{\cal C}_k=\{(x,u)\in \R\times [0,t-T^\nu]\, |\, |x+(1-2\rho_k)(t-u)-X|< T^{2/3+\e}\},
\end{equation}
for $k=1,2,3$. By choosing $2/3+\e<\nu<1$, we clearly have ${\cal C}_1\cap {\cal C}_2 \cap {\cal C}_3=\emptyset$ for all $T$ large enough.

Let $\bar x^k_{N_k(u)}(u)$ be the backwards process starting at time $t-T^\nu$ with particle number $N_k(t-T^\nu)=N-\rho_k^2 T^\nu$. Then, rewriting the result of Theorem~\ref{thmLocalization} for the three densities and shifting appropriately, we get
\begin{cor}\label{corLocalization}
Let $0<\e<1/3$ and $2/3+\e<\nu<1$. Then for all $T$ large enough,
\begin{equation}
\Pb((\bar x^k_{N_k(u)}(u),u)\in {\cal C}_k\textrm{ for all }u\in[0,t-T^\nu])\geq 1-C e^{-c T^\e},
\end{equation}
for $k=1,2,3$.
\end{cor}

By the construction of the backwards path, on the event that the path stays inside the cylinder ${\cal C}_k$, $z^k=\bar x^k_{N_k(t-T^\nu)}(T-T^\nu)$ does not depend on the randomness outside ${\cal C}_k$, for $k=1,2,3$. For $T$ large enough, the three cylinders are disjoints. Thus, on the event that all three backward paths stay in the respective cylinder, the random variables $z^1,z^2,z^3$ are independent. Since by Corollary~\ref{corLocalization} this event has probability going to one as $T\to\infty$, it means that $z^1,z^2,z^3$ are asymptotically independent and thus
\begin{equation}
(\ref{eq3.15})=\prod_{k=1}^3  \lim_{T\to\infty}\Pb(z^k\geq X-s T^{1/3}).
\end{equation}
The r.h.s.~is given in Corollary~\ref{CorCvgOnePtDistr}, thus ending the proof of Theorem~\ref{ThmMain}.

\appendix

\section{Some bounds}

\begin{lem}\label{lemBound2}
Let $\nu \in (0,1)$. There exists a $t_0\in (0,\infty)$ such that for all $t\geq t_0$,
\begin{equation}\label{BoundsPtPt}
\begin{aligned}
&\Pb(x^{\rm step}_{\nu t}(t)\geq (1-2\sqrt{\nu})t-s t^{1/3})\leq C_1\, e^{-c_1 (-s)^{3/2}},\quad s\leq 0,\\
&\Pb(x^{\rm step}_{\nu t}(t)\leq (1-2\sqrt{\nu})t -s t^{1/3})\leq C_2 \,e^{-c_2 s},\quad s\geq 0,
\end{aligned}
\end{equation}
where the constants $C_i,c_i$ are positive and independent of $s$. Further, for any given $\e>0$, the constants in the bounds for step initial conditions can be chosen independent of $\nu\in [\e,1-\e]$.
\end{lem}
The first estimate in (\ref{BoundsPtPt}) was obtained in~\cite{BFP12} in terms of TASEP height function. The idea is to bound the Fredholm determinant which gives the distribution function of $x^{\rm step}_{t/4}(t)$ by the exponential of the trace of the kernel, see Section~4 of~\cite{BFP12}. The method was used before by Widom in~\cite{Wid02}. The other two estimates in (\ref{BoundsPtPt}) follow directly from the exponential estimates on the correlation kernel for step initial condition. Although we do not need it here, let us mention that the estimates can be improved to optimal decay power~\cite{LR09}.


\end{document}